\documentclass[10pt,letterpaper]{article}
\usepackage[top=0.85in,left=2.75in,footskip=0.75in,marginparwidth=2in]{geometry}
\usepackage{graphicx}
\usepackage{mathptmx}      
\usepackage{latexsym}
\usepackage{amsmath}
\usepackage{amssymb}
\usepackage{amsthm}

\newcommand{\gs}{{\bf{s}}}
\newcommand{\gx}{{\bf{x}}}
\newcommand{\gp}{{\bf{p}}}
\newcommand{\gt}{{\bf{t}}}

\newcommand{\MFM}{{\mathfrak M}}

\newcommand{\MCE}{{\mathcal{E}}}

\newcommand{\Z}{\mathbb{Z}}

\newtheorem{proposition}{Proposition}
\newtheorem{theorem}{{Theorem}}

\newtheorem{corollary}{{Corollary}}
\theoremstyle{remark}
\newtheorem{definition}{{\bf Definition}}

\begin{document}


%
%

\title{A Linear Algorithm For Computing Polynomial Dynamical Systems}

%
%
%

\maketitle

Ines Abdeljaoued-Tej\textsuperscript{1,*}, 
Alia Benkahla\textsuperscript{2}, 
Ghassen Haddad\textsuperscript{3}, 
Annick Valibouze\textsuperscript{4}

\bigskip
\noindent
{\small 
{\bf{1} University of Tunis El Manar, Laboratory of BioInformatics, bioMathematics and bioStatistics (BIMS) in Institute Pasteur of Tunis, 1002 Tunis and University of Carthage, ESSAI; Tunisia}
\\
{\bf{2} Laboratory of BioInformatics, bioMathematics and bioStatistics (BIMS), Institute Pasteur of Tunis,
University of Tunis El Manar; Tunisia }
\\
{\bf{3} University of Tunis El Manar, Laboratory of BioInformatics, bioMathematics and bioStatistics (BIMS), Institute Pasteur of Tunis and ENIT; Tunisia \\ Sorbonne University, Laboratory Jacques Louis Lions (LJLL); France }
\\
{\bf{4} Sorbonne University, Laboratoire d'Informatique de Paris 6 (LIP6) and Laboratoire de Statistique Th\'eorique et Appliqu\'ee (LSTA); France }
\\
}
\bigskip
* inestej@gmail.com

\begin{abstract}
Computation biology helps to understand all processes in organisms from interaction of molecules to complex functions of whole organs. Therefore, there is a need for mathematical methods and models that deliver logical explanations in a reasonable time. For the last few years there has been a growing interest in biological theory connected to finite fields: the algebraic modeling tools used up to now are based on Gr\"obner bases or Boolean group. Let $n$ variables representing gene products, changing over the time on $p$ values. A Polynomial dynamical system (PDS) is a function which has several components; each one is a polynom with $n$ variables and coefficient in the finite field $\Z/p\Z$ that model the evolution of gene products. We propose herein a method using algebraic separators, which are special polynomials abundantly studied in effective Galois theory. This approach avoids heavy calculations and provides a first Polynomial model in linear time. 
\end{abstract}

{\noindent\it Keywords: Polynomial dynamical system; Finite field; Mathematical model; Gene regulatory network.\\
Code: 37-XX, 37-E-15, 12E20, 12E05, 11Txx, 92C42, 37N25, 12F10, 12Y05, 11Y40}

\section{Introduction}

A living cell can be compared to a complex factory animated by DNA and large amounts of biological data are now available. Therefore, the challenge is to extract biological meaning: it consists of developing methodologies for using these data to address biological questions \cite{doi:10.1093/bioinformatics/bti565}. 
System biology is a study of several networks related to biology data. We choose gene expression data to design regulatory gene network and to build a mathematical model from observations of the system response to well-constructed perturbations. The data are measurements of concentrations of biochemicals, recorded at time intervals: we propose an adjusted model, defined in a finite field, that fits these data. \\

Several methods have been proposed to infer gene interrelations from expression data. To solve this issue, many researchers have proposed various approaches \cite{Schilstra2004}. In fact, mathematical network modelling is an important step toward covering the dynamic behaviour of biological networks. New areas of mathematics, not traditionally considered applicable, are now contributing with powerful tools. Some applied research were focused on investigating dynamical properties. They propose original analysis for regulatory interactions where Boolean models were generated \cite{Paroni2016,hink}. A linear system of differential equations, obtained from measured gene expression, can be used to infer gene regulatory network \cite{ed}. Bayesian networks are used to measure gene expression data \cite{imoto}, as well as modeling gene expression data \cite{Friedman00usingbayesian}. \\

Polynomial dynamical system is an approach used for understanding the behavior of complex systems over time \cite{thomas:hal-00087681}. All these approaches deal with biological systems which studies and describes the interactions between micro-biological outputs. It finds its roots in symbolic computation and mathematical modelling. One of the precursors of Polynomial dynamical system is R. Thomas with his Boolean dynamical system \cite{thomas}. In the last decade, several studies have been made, including contributions in Hybrid systems biology \cite{Bortolussi2008}. The main objective of the paper is to define an algorithm which computes biological systems, over finite fields, in linear time. Our approach is part of a broader framework of a multidisciplinary team working on an algebraic modelling as well as based on EDO or Bayesian networks. With this in mind, we started in \cite{Aut2008} by presenting a summary of the main methods adapted to this framework. The algebraic method presented in this paper adapts the techniques of Galois theory to issues of bio-systems. This approach can effectively model the important size of the biological data, with relatively simple tools: the calculation of elementary symmetric polynomials in the case of boolean or the fundamental modules in the polynomial case. \\

The first use of Polynomial dynamical system (PDS) on System biology was published by R. Laubenbacher and B. Stigler \cite{lauben03,shortcourseStigler}. Their model is a deterministic graphical model which depends on the degree of data discretization $p$ ($p=2$ in boolean modelling). For example, their model describes the relationships between diseases and symptoms by producing dependency graphs (also known as the wiring diagram \cite{laubenstigler}) and space state graphs. Given symptoms, the space-state graph can be used to compute the presence of various diseases. Thus, when $n$ entries evolve in time, the dynamical function can be represented by $n$ polynomials which describe a table of $p^n$ possible state. Our researches were inspired by a classical method based on Lagrange interpolation \cite{Lagrange:1770}. We propose herein two methods using algebraic separators, both are special polynomials abundantly studied in effective Galois theory. Algebraic separators are directly determined by using symmetric functions, or by linear combinations of fundamental modules. These approaches avoid heavy calculations of Gr\"obner bases and provide a first Polynomial model in linear time, similar to the one produced by the team of R. Laubenbacher. The computation of a first Polynomial dynamical system, as detailed in this article, can perfectly be computed in parallel. Tools based on differential equations, Bayesian networks or Boolean networks are commonly used to model biological networks \cite{cinquemani:hal-01399942}. Unlike these techniques and despite the important work done within R. Laubenbacher's team \cite{jarrah-2006}, the Polynomial modelling remains underexploited. Very few publications are available in the literature that addresses the issue of algebraic modelling in biology. Because of the large amount of information present in a PDS, other network forms can easily be derived from it. Further work is needed to discern models that balance these needs optimally. Issues regarding model personalization and boundary condition tuning are particularly important. Instead, we use computational algebra to discuss the biochemical networks using the example of gene regulatory networks. \\

We propose a linear algorithm that performs learning in Polynomial dynamical systems. The paper offers an effective alternative for Gr\"obner basis when modeling biological systems. We introduce a novel method of modelling based on polynomial over finite fields. The remainder of the paper is organized as follows: Section \ref{sec:00} details an approach allowing separators computation: we start by introducing useful definitions and notations. Section \ref{sec : 2} presents a method based on Galois theory tools as the {\it fundamental modules} or {\it elementary symmetric functions}. Section \ref{sec2:1} is devoted to introduce a main theorem for computing algebraic separators with optimal degree; it allows us to compute a finite number of other Polynomial dynamical systems. Section \ref{sec2:0} summarizes the particular case of Boolean dynamical Systems. Experimental results are given in Section \ref{aff}: some basic techniques that enable us to compute affine separators. Section \ref{sec:443} gives an algebraic technique for defining different types of rules between genes. We illustrate our algorithm with a network where states are in a finite field. The main feature of the data generated by DNA microarrays, called {\it gene expression data}, is that they have few experiments with regard to the number of genes tested simultaneously. Unfortunately, it is difficult to get all the data's necessary information to verify if the model is valid. In order to propose a mathematical model based on the polynomial dynamical systems, we discretize the data by using a double filter for changing genes expressions. The data we have chosen to study, are divided in two lines for two different media. Taking into account the expression data of genes over time, we discusses an example of computing a Polynomial dynamical system for each cell line computed in linear time. Finally,  Section \ref{conc} gives the conclusion with the perspective.

\section{Polynomial dynamical system}
\label{sec:00}

In the following, let $p$ be a prime number and $k$ be a finite field with $p$ elements: $k$ is also noted as $\Z/p\Z$. Let $\mathcal{E}=k^n$ be the set of $p^n$ states.
A {\it Polynomial dynamical system of dimension $n$}, denoted by {PDS}, is a polynomial function $f = (f_1,\dots ,f_n )$ whose components $f_i$ are polynomials of the quotient ring $k\lbrack x_1, \dots ,x_n\rbrack / <x_1^p-x_1,\dots , x_n^p-x_n>$, i.e., polynomials on $n$ variables with coefficients in $k$ and degree smaller than $p - 1$ in each variable. \\

For example, the data commonly used in the modelling of biological systems are concentrations of $n$ proteins and the change of these concentrations over $m$ time points.
We usually need normalized data but even if normalization is sufficient for continuous models, it is inappropriate for discrete models, as it is for Boolean or Polynomial Networks.
Indeed, in the latter cases, a discretization is also necessary. It can be done using E. Dimitrova's method \cite{DBLP:journals/jcb/DimitrovaVML10}. The latter gives an integer $p$, called degree of discretization, and replaces each value by $x_j$ with $j\in [1,n]$, a corresponding value in $\Z/p\Z$.
At each time $t\in \lbrack1,m\rbrack$, the vector $\gs_t=(x_1,x_2,\dots,x_n)$ is called {\it the state of the system at time $t$}. The finite trajectory $\gs_1\mapsto\gs_2\mapsto\dots\mapsto\gs_m$ is called {\it a discrete trajectory of length $m$}. We consider a PDS, which is a polynomial function $f=(f_1,\dots,f_n)$ satisfying, for time $i\in[1,m-1]$: \begin{eqnarray}f(\gs_i)=\gs_{i+1}.\end{eqnarray}

There is an infinite number of such PDS because it is a solution of a multivariable interpolation problem. 
In fact, let $I(V)=\lbrace h\in k[x_1,\dots ,x_n]\mid h(\gs)=0, \forall \gs\in V\rbrace$ be the ideal of affine variety $V\subset {\mathcal{E}}$, $j\in [1,n]$. Let $j\in [1,n]$. For each $i\in[1,m-1]$, $f_j(\gs_i)=\gs_{i+1,j}$. For every $g\in k[x_1,\dots ,x_n]$ there exists $h\in I(V)$ such as $g=f_j+h$. The polynomial $f_j-g\in I(V)$ verifies also $g(\gs_i)=\gs_{i+1,j}$ for $i\in[1,m-1]$.

It is possible to separate an element of $V$ with respect to others thanks to a polynomial of $k[x_1,\ldots ,x_n]$: \begin{definition}
\label{def:0}
Let $V\subset\MCE$ and $\gs$ a state in $V$. A polynomial $r(\gx)\in k[x_1,\ldots ,x_n]$ which equals 1 at $\gx=\gs$ and 0 in the other elements of $V$ is called a (polynomial) {\it separator of $\gs$ in $V$}.
\end{definition}
Separators have the same purpose as Lagrange's polynomials in a univariate interpolation problem. For example, following the results of \cite{shortcourse} and taking into account a Gr\"obner basis's computation of the ideal $I(V)$, we obtain separators $r_j$ of $\gs_j$ in V for $j\in [1,m-1]$:
\begin{eqnarray}
r_1(x_1,x_2,x_3)&=&-2x_1x_3 - 2x_2x_3 + 2x_3^2\\
r_2(x_1,x_2,x_3)&=&x_2^2 -2x_2x_3 +2x_3^2\\
r_3(x_1,x_2,x_3)&=&-2x_1 +2x_2+x_3\\
r_4(x_1,x_2,x_3)&=&-x_2x_3 - 2x_3^2.
\end{eqnarray}
We compute a PDS $f=(f_1,\dots ,f_n)$ by interpolation resolution based on separators: starting from the respective concentrations $(4,3,1,0)$, $(4,3,2,0)$ and $(0,1,3,4)$ of  proteins $1$, $2$ and $3$, we compute a PDS verifying $f_j(\gs_i)=\gs_{i+1,j}$ for $j=1,2,3$ and $i=1,2,3,4$:
{\small
\begin{eqnarray}\label{df}
f_1(x_1,x_2,x_3)=4r_1+3r_2+1r_3+0r_4&=&3 x_2^{2} + 2 x_1x_3 + x_2x_3 - x_3^{2} + 3 x_1 + 2 x_2 + x_3 \\ \label{df2} f_2(x_1,x_2,x_3)=4r_1+3r_2+2r_3+0r_4&=&3 x_2^{2} + 2 x_1 x_3 + x_2x_3 -x_3^{2} + x_1 - x_2 + 2 x_3\\ \label{df3}f_3(x_1,x_2,x_3)=0r_1+1r_2+3r_3+4r_4&=&x_2^{2} - x_2x_3 - x_3^{2} - x_1 + x_2 + 3 x_3 
\end{eqnarray}}

\begin{figure}[h]
\label{Fig12}
   \includegraphics[width=14cm,height=12cm]{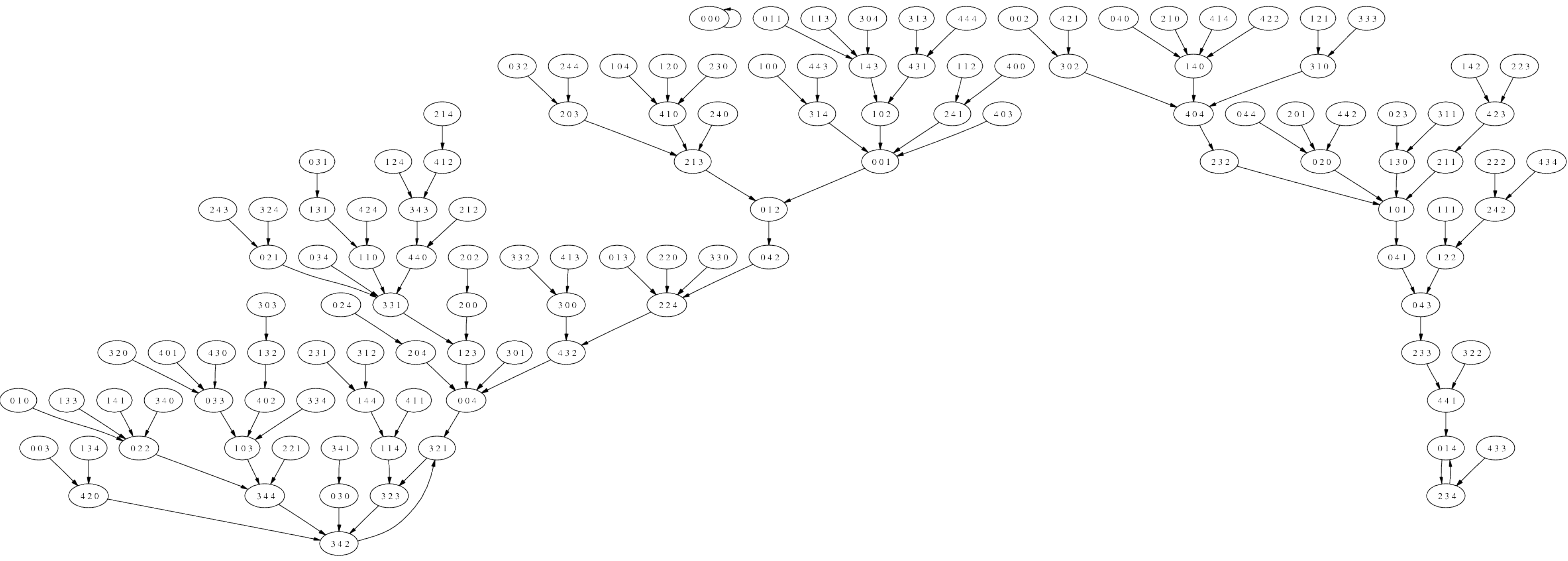}
   \caption{Space state Graph ($n=3$ and $p=5$) with respect to (\ref{df}), (\ref{df2}) and (\ref{df3})}
\end{figure}
From $f=(f_1,f_2,f_3)$, we can deduce all the trajectories using DVD developed in \cite{jarrah-2005}. This example illustrates a Polynomial Networks reproduced by the Buchberger-M\"oller algorithm which is implemented in Computer Algebra Systems like Macaulay2 \cite{M2} or Cocoa \cite{CoCoALib}. 
R. Laubenbacher and B. Stigler had discussed the Lagrange interpolation without pushing the analysis to develop a specific method for computing algebraic separators. However, they adapted Gr\"obner's basis computation of $I(V)$ to all PDS, and in particular, to obtain algebraic separators in $o(n^2m^3+nm^2)$. 

\subsection{Algebraic Separators}
\label{sec : 2}

Let ${\gs} = (s_1,\dots ,s_n) \in \mathcal{E}$ where $\mathcal{E}=k^n$ is the set of $p^n$ states.  
In $k[x_1,\ldots ,x_n]$, the maximal ideal $\MFM_{(\gs)}$ of $\gs$-relations is generated by 
\begin{eqnarray}p_1=x_1-s_1,\quad p_2=x_2-s_2,\quad\ldots\quad  p_n=x_n -s_n\, .\end{eqnarray}
The set $\gp=\{p_1,\dots ,p_n\}$ is called by N. Tchebotarev the set of {\it fundamental modules of $\gs$} \cite{Tchebotarev:50}. The fundamental module of $\gs$ satisfies for all $\gt\in \MCE$:
\begin{eqnarray}
\forall j\in[1,n] \quad p_j(\gt) = 0 \quad \Leftrightarrow \quad \gt=\gs \quad .
\end{eqnarray}
Which means that the variety of $\MFM_{(\gs)}$ is reduced to the single element $\gs$. It is a special case of Galois theory in which the Galois group over the field $k$ of the polynomial $(x-s_1)\cdots (x-s_n)$ is the identity group \cite{Valibouze:08}. Let $J$ be the set of indices $j$ for which all elements of $V=\{\gs_1,\gs_2,\dots ,\gs_m\}$ have the same $j$-th coordinate:
$
J=\{j\in[1,n]  \mid \forall l \in [1,m]\; s_{l,j}=s_{1,j} \}$;
On the coordinates indexed by $j$ no separation is possible. Fix $S=\{1,\ldots ,n\} \backslash J$, the minimum subset of $\{1,\ldots ,n\}$ which separates $V$'s elements, $S$ is called {\it separator set of $V$}.
Note that the set $J$ excludes the gene products (like proteins) whose concentration is constant in time. We keep our calculations for genes products that vary over time: these are $S$'s elements. In particular, for $V=\MCE$, the separator set is $S=\{1,\ldots ,n\}$. Let us consider this first theorem that lights us on the algebraic computation of separators:

\begin{theorem}\label{theo : 3}
Let be the univariate polynomial 
\begin{eqnarray}
g(\gx)=\prod_{j\in S} \prod_{l\in E}(p_j(\gx) -l)
\end{eqnarray}
where $S$ is the separator set of $V$ and $E$ all non-zero values taken by the points of $V$ on the  ideal's generators of $\gs$-relations (except for those unnecessary to separation): 
\begin{eqnarray}
E=\{p_j(\gt) \mid j\in S \;; \;  p_j(\gt)\neq 0 \;; \; \gt \in V  \} \subset \{1,\ldots ,p-1\}.
\end{eqnarray}
Then the polynomial 
\begin{eqnarray}
r(\gx)=\frac{g(\gx)}{g(\gs)}
\end{eqnarray}
is a separator of $\gs$ in $V$.
\end{theorem}

\begin{proof}
Let $g(\gs)\neq 0$ and so $r(\gs)=1$. In fact, for all $j\in S$ and $l\in E$, $p_j(\gs)-l=-l\neq 0$. The product $\prod_{l\in E}-l\neq 0$ in $k$ because $E\subset  \{1,\ldots ,p-1\}$ and $p$ is a prime number. Let $\gt \in V\backslash \{\gs\}$, then it exists $j\in[1,n]$ where $p_j(\gt)\neq 0$ ; by definition of $E$, $l=p_j(\gt)\in E$ ;  as the $g$'s factor $p_j(\gx)-l\neq 0$ for $\gx=\gt$, we obtain the identity $r(\gt)=g(\gt)=0$~;
\end{proof}

\begin{proposition}
The complexity of calculation proposed in Theorem \ref{theo : 3} equals $ o (n(p-1)) $. 
\end{proposition}

\begin{proof}
Let $card(E)$ be the cardinal of the set $E$ and $card(S)$ the cardinal of the set $S$. In the worst case, computing the univariate polynomial 
$
g(\gx)=\prod_{j\in S} \prod_{l\in E}(p_j(\gx) -l)
$ uses $ card(S) = n $ and $card(E) = p-1$. On the other hand,  $ p_j (x) = x_j-s_j$ and computing a separator $ r (x) $ needs $ 2n (p-1) $ additions, $ (n-1)(p-1)$ multiplications and 1 division in $ \Z / p\Z $.  
\end{proof}

To obtain a PDS associated with a discrete path of length $ m $, it suffices to compute $ m-1 $ separators. The result is not optimal (see Corollary \ref {theo:theo}) but the computational complexity remains linear in $ n $, which is very interesting to find a first study model (remember that biological data are such that $ p $ and $ m $ are very small compared to $ n $). 

\subsection{Optimal Algebraic Separators}
\label{sec2:1}

In the previous section, and specially in Theorem \ref{theo : 3}, we separate repeatedly the same state $\gs$ from a state $\gt$ in $V$. To avoid this redundant operation, let us reduce the degree of $g$. For $\gt\in V$ and $j\in[1,n]$, $p_{\gt,j}(\gx)=x_{j}-t_{j}$. 

\begin{definition}
A {\it separate set of two distinct points $ \gs $ and $ \gt $ of $ V $} is given by:
$$
S(\gs,\gt)=\{j \in S\mid s_{j}\neq t_{j}\}=\{j \in S\mid p_{\gs,j}(\gt)\neq 0\}
$$
where $S$ is the separator set of $V$. The {\it initial set of separators of $\gs$ and $\gt$} is:
$$
PS(\gs,\gt)=\{ p_{\gt,j}(\gx)  \mid j \in S(\gs,\gt)\}.
$$
A \textit{minimal initial separator's set of $\gs$ in $V$}, denoted by $Min(\gs,V)$ is composed of less than $m-1$ distinct polynomials $g$ the same applies to each $\gt \in V$ distinct from $\gs$, the intersection of $Min(\gs,V)$ with $PS(\gs,\gt)$ is reduced to a single element. 
\end{definition}

For every distinct point $\gs$ and $\gt$ in $V$ and for each $g \in PS(\gs,\gt)$, we have $g(\gt)=p_{\gt,j}(\gt)=0$ and $g(\gs)=p_{\gt,j}(s)=-p_{\gs,j}(t)\neq 0$ because $j\in S(\gs,\gt)$.
Any product of initial separators of $ \gs $ and $ \gt $ is a separator between these two points. 

\begin{corollary}
\label{theo:theo}
Let $Min(\gs,V)$ a minimal initial separator's set of $\gs$ in $V$ and let 
\begin{eqnarray}
{\mathcal G}=\prod_{g \in Min(\gs,V)} g
\end{eqnarray}
then 
\begin{eqnarray}
r(\gx)=\frac{{\mathcal G}(\gx)}{{\mathcal G}(\gs)}
\end{eqnarray}
is a separator of $\gs$ in $V$. 
\end{corollary}

We put $r_s=r$ or $r_i =r$ when $\gs$ is indexed by $i$. To illustrate the above definitions, consider the statements
$\gs_{1}=(2,1,2)$, $\gs_{2}=(1,1,0)$, $\gs_{3}=(0,0,1)$ and $\gs_{4}=(1,2,0)$. The separator set is $S=\{ 1,2,3\}$ and the respective separator sets of pair of elements of $V$ are: $
S(\gs_{1},\gs_{2})=\{1,3\}, S(\gs_{2},\gs_{4})=\{2\}$, 
$ S(\gs_{1},\gs_{3})=S(\gs_{1},\gs_{4})=S(\gs_{2},\gs_{3})=S(\gs_{3},\gs_{4})=\{1,2,3\}$. 

We get $\gs_1$'s initial sets of separators:
$
PS(\gs_{1},\gs_{2})=\{x_1-1,x_3\}$, 
$PS(\gs_{1},\gs_{3})=\{x_1,x_2,x_3-1\}$, 
$PS(\gs_{1},\gs_{4})=\{x_1-1,x_2-2,x_3\}$. Here are some minimal initial separator's sets $Min(\gs_1,V)$ of $ \gs_1$ in $ V $ (there is a finite number):
$
\{x_1-1,x_1\}, \{x_1-1,x_2\}, \{x_3,x_2\}$.
We compute a separator $ r_1 $ of $ \gs_1 = (2,1,2) $. From the minimum set $ \{x_3,x_2\} $ of $ \gs_1 $ in $ V $, we obtain the polynomial
$
r_{1}={x_{2}x_3}/{2}=-x_2x_3
$
which equals to 1 in $\gs_{1}$ and 0 everywhere in $V$. This is the same as the separator in \cite{laubenstigler}. A separator of $s_4$ obtained by Corollary \ref{theo:theo} is 
$
r_{4}=(x_{1}-2)(x_2-1)(x_3-1)=x_1x_2x_3 - x_1x_2 - x_1x_3 + x_2x_3 + x_1 - x_2 - x_3 + 1$. \\

The computation of each separator can be done simultaneously \cite{iaa}: in fact, the separators of $V$'s elements are independent from each other. So, the computation of a PDS can be done in parallel by applying for each separator the Theorem \ref{theo : 3}. 
Note that the separators obtained by Corollary \ref{theo:theo} have degree $ \leq n $, while those obtained by Theorem \ref{theo : 3}  have degree $ \leq n.card(E) $. For example, for $ n = 3$ and $ E=\{ 1,2\}$, $ n .card(E) = 6$ and the $\gs_1 $ separator obtained by Theorem \ref{theo : 3}  (with linear complexity) is $ r_{1}=-x_1(x_1-1)(x_2-2)x_2x_3(x_3-1)$.

\subsection{Boolean Dynamical Systems}
\label{sec2:0}

We suppose, for simplicity, that the separator's set $S$ of $V$ is $\{1,\dots ,n\}$. In the case $ p = 2$, $ E $ is reduced to $\{ 1\}$ and separators are in a more compact form. In fact, let $ k = \Z / p \Z $ and let two polynomials in $ k [x_1, \ldots, x_n] $ given by:
\begin{eqnarray*}
q&=&\sum_{r=1}^n e_r(\gp) \quad  \text{ and}\\
r&=& q+1 
\end{eqnarray*}
where $e_1(\gp),\ldots ,e_n(\gp)$ are elementary symmetric functions in the elements of $\gp=\{p_1,\dots ,p_n\}$. To compute, $e_i(\gp)$, we use tools developed in \cite{val:sym}. For all $\gt \in \MCE$
\begin{eqnarray*}
q(\gt)=0 &\Leftrightarrow& \gt = \gs\\
r(\gt)=0 &\Leftrightarrow& \gt \neq \gs.
\end{eqnarray*}
In particular, $q(\gt)=1 \quad \forall \gt\neq\gs$ and $r(\gs)=0$. To see these properties on $ q $ and $ r $, we consider the polynomial $ g $ in univariate $ y $ with coefficients in $ k [x_1, \ldots, x_n] $ whose roots are the fundamental modules, being written as follows:
$$
g(y)(\gx) = \prod_{j=1}^n (y-p_j(\gx)).
$$
The polynomial $ g (1) (x) $ is the same as that of Theorem \ref{theo : 3}.
For $t\in \MCE$, we have the following equivalences:
$$
g(1)(\gt)= 0 \Leftrightarrow \exists j \in[1,n] \; : p_j(\gt) =1 \Leftrightarrow \gt\neq \gs.
$$
So $ g (1) $ is a polynomial separating for $ \gs $ in $\MCE$. Moreover, according to the following identity on the coefficients of univariate polynomials:
$$
g(y)= y^n -e_1(\gp)\cdot y^{n-1}+e_{2}(\gp)\cdot y^{n-2} \cdots + (-1)^n\cdot e_{n}(\gp),
$$
we have $r=g(1)=q+1$ in $\Z/2\Z$. \\

To illustrate this result, let $m=3$, $p=2$, $\gs_1=\left(0,1, 1\right)$, $\gs_2=\left(1, 
 1, 
 0\right)$ and $\gs_3=\left(0, 
 1, 
 1\right)$. The fundamental module sets associated to each state $\gs_i$ is denoted by $\gp_i$; we obtain: 
\[\gp_1=\{ x_{1}, 
 x_{2} + 1, 
 x_{3} + 1\}\quad and\quad\gp_2=\{x_{1} + 1, 
 x_{2} + 1, 
 x_{3}\}.\] The elementary symmetric functions of $\gp_1$ are: \\
$e_1(\gp_1)=x_{1} + x_{2} + x_{3}
$, $e_2(\gp_1)=
x_{1} x_{2} + x_{1} x_{3} + x_{2} x_{3} + x_{2} + x_{3} + 1$ and $e_3(\gp_1)=
x_{1} x_{2} x_{3} + x_{1} x_{2} + x_{1} x_{3} + x_{1}$ ; so a separator of $\gs_1$ is \[ r_1(x_1,x_2,x_3)=x_{1} x_{2} x_{3} + x_{2} x_{3}.\]
In parallel, we find the separator $r_2$ of $\gs_2$: 
\[ r_2(x_1,x_2,x_3)=x_{1} x_{2} x_{3} + x_{1} x_{2}.\]
For $t_1=(1,0)$, $t_2=(1,1)$ and $t_3=(0,1)$, we obtain $f_1=r_1, f_2=r_1+r_2$ and $f_3=r_2$. So, $f=(f_1,f_2,f_3)$ is defined by: \begin{equation}\label{eq:11}
f_1=x_{1} x_{2} x_{3} + x_{2} x_{3}\quad,\quad f_2=x_{1} x_{2} + x_{2} x_{3}\quad{\rm and}\quad f_3=x_{1} x_{2}
x_{3} + x_{1} x_{2}.
\end{equation}

The base of each component of the PDS $ f $ is used to draw the dependency graph (shown in Figure 2). Note that using Cocoa, we must consider the states $ \gs_1 $ and $ \gs_2$ that give $ r_1 = x_3 $ and $ r_2 = x_2 + x_3 $. Hence, taking into account $\gs_3 $: $ f = (x_3, x_2, x_2 + x_3) $. This Boolean dynamical system allows us to obtain the dependency graph in Figure 3 (graph which includes a smaller number of arcs than the PDS (\ref{eq:11})). 
The dependency graph is a directed graph with vertex set $\lbrace
x_1,...,x_n\rbrace$ and edge set $\lbrace (t, x_i) \mid t \in supp(f_i), i = 1,...,n\rbrace$ where the support of $f_i$, denoted $supp(f_i)$, is the set of variables
that appear in $f_i$.

    \begin{figure}
        \includegraphics[scale=0.6]{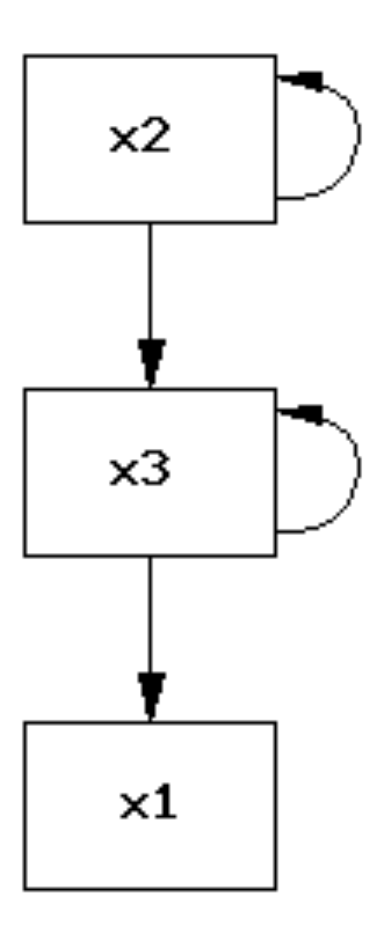}
        \caption{Dependency graph of $f=(x_3,x_2,x_2+x_3)$ where $n=3$, $m=3$ and $p=2$}
        \label{44}
    \end{figure}
    ~\qquad\qquad 
    \begin{figure}
        \includegraphics[scale=0.6]{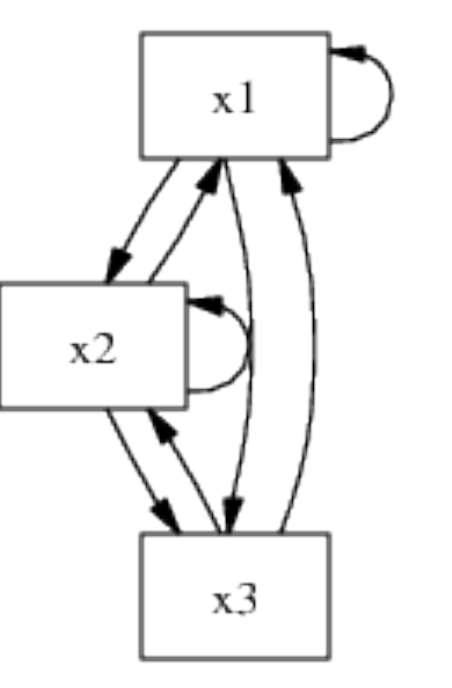}
        \caption{Dependency graph of $f=(x_{1} x_{2} x_{3} + x_{2} x_{3},x_{1} x_{2} + x_{2} x_{3},x_{1} x_{2}x_{3} + x_{1} x_{2})$ where $n=3$, $m=3$ and $p=2$}
        \label{33}
    \end{figure}

\subsection{Affine Separators}
\label{aff}

We detail in this section a method for determining, in some cases, affine separators. Note 
$
CL(\gt)
$
the set of all linear combinations of generators and $ \MFM_ {(\gt)} $ the ideal of 
$ (\gt) $-relations. To find affine separators of $ \gs $ in $ V $, we must determine
$$
(\cap_{\gt \neq \gs}CL(\gt) )\backslash CL(\gs).
$$
This set is a set of affine separators of $ \gs $ in $ V $. To determine if a relationship does not belong to $ \MFM_ {(\gs)} $ (i.e., to $ CL (\gs) $), simply divide in the fundamental modules of $ \MFM_{ (\gs)} $ (or what amounts to even evaluate $ \gs $). Generally \[r_j\in\cap_{i \neq j}\MFM_{(\gs_i)}\quad {\rm and}\quad r_j\notin\MFM_{(\gs_j)},\] for any separator $r_j$ of $\gs_j$ in $V$. \\

Let 
$\gs_{1}=(2,1,2)$, $\gs_{2}=(4,4,0)$, $\gs_{3}=(3,3,1)$, $\gs_{4}=(1,2,3)$ ; $S=\{1,2,3,4\}$. First, we add the generators 2-by-2 ($p_{i,j}+p_{i,k}$), second we subtract ($p_{i,j}-p_{i,k}$) and $2p_{i,j}+p_{i,k}$, etc. Then we pass $p_{i,1}+p_{i,2}+p_{i,3}$ ; for example
{\small 
\begin{eqnarray*}
\MFM_{(\gs_1)}&=&<x_{1} + 3, 
 x_{2} + 4, 
 x_{3} + 3>\\
CL(\gs_{1})&=&\{0,2 x_{1} + 3 x_{2} + x_{3} + 1, 3 x_{1} + 2 x_{2} + 3 x_{3} +
1, 3 x_{1} - x_{2} + x_{3} + 3, \dots , \\ & &
x_{2} - x_{3} + 1, x_{1} + x_{2} + 2 x_{3} + 3, 3 x_{1} + 2 x_{2} + 2, \dots, {\bf 2x_1 + 2x_2 - 1, }\ldots \} \\
\MFM_{(\gs_2)}&=&<x_{1} + 1, 
 x_{2} + 1, 
 x_{3}>\\
CL(\gs_{2})&=&\{ - x_{2} + 4, - x_{1} - x_{2} - x_{3} + 3, 3 x_{1} + 2 x_{2} +
3 x_{3}, 3 x_{1} - x_{2} + x_{3} + 2, \\ & &- x_{1} - x_{3} + 4, 3 x_{1} +
x_{3} + 3, 2 x_{1} + x_{2} + x_{3} + 3, x_{1} + x_{2} + 3 x_{3} + 2, \\ & & 3
x_{1} + x_{2} + 2 x_{3} + 4, x_{1} + 3 x_{2} + x_{3} + 4\quad {\bf 2x_1 + 2x_2 - 1}, 2 x_{1} + 2 x_{2} + 3 x_{3} + 4,\ldots \} 
\end{eqnarray*}}
{\small \begin{eqnarray*}
\MFM_{(\gs_3)}&=&<x_{1} + 2, 
 x_{2} + 2, 
 x_{3} + 4>\\
CL(\gs_{3})&=&\{ - x_{2} + 3, 3 x_{1} + 3 x_{2} + 2, 3 x_{1} + 2 x_{2} + 3
x_{3} + 2, 3 x_{1} - x_{2} + x_{3} + 3,  \\ & & \dots ,x_{1} - x_{2} + 3 x_{3} + 2, x_{1} + 3 x_{3} + 4\} \\
\MFM_{(\gs_4)}&=&<x_{1} + 4, 
 x_{2} + 3, 
 x_{3} + 2
>\\
CL(\gs_{4})&=&\{ - x_{1} - x_{2} - x_{3} + 1, 3 x_{1} - x_{2} + x_{3} + 1, -
x_{1} - x_{3} + 4, 3 x_{1} + 2 x_{2} + 3 x_{3} + 4, \\ & & {\bf 2x_1 + 2x_2 - 1}, x_{1} + 3
x_{3}, 3 x_{1} + 3 x_{2} + x_{3} + 3, - x_{2} + 2 \ldots \}
\end{eqnarray*}}
We quickly find the separator $2x_1 + 2x_2 - 1$ of $\gs_3$ and we can also find a linear identification. $r_3=2x_1 + 2x_2 - 1$ is an element of $(\cap_{i \neq 3}CL(\gs_i) )$ that does not belong to $CL(\gs_3)$. In this example, it is not possible to find affine separators of $ \gs_1 $, $ \gs_2 $ and $ \gs_4$. \\

This method can fully be treated dynamically and there is no question to compute full $ CL (\gt) $. In fact, there is an interesting algorithm for implementation. Several strategies are possible including parallel computing. On the other hand, to avoid repeating the same relationship in the ideal of $ (\gt) $-relations $ \MFM_ {(\gt)} $, we choose a standard format for each polynomial. For example, the polynomial is the same denominator (which we withdrew), it is divided by its content (the gcd of its coefficients) and multiplying the coefficient of the smallest variable by its sign to make it positive. The authors of \cite{shortcourse} explain that the order of $ x_i $ can be predetermined by information concerning the influence of certain genes over others: e.g. $x_{1}<x_{2}<\cdots <x_{n}$. \\
Thus, we obtain a first model in the form of a PDS with linear complexity in the order of $ o (n. (p-1). (m-1)) $. Many other PDS can be determined with our approach optimally. In the other hand, getting all Polynomial dynamical systems from one PDS computed algebraically is also possible: a calculation of the generators of the ideal $I(V)=\{h\in k[\gx]\quad \mid h(\gs)=0\quad ,\forall \gs\in V\}$. The computation of all the PDS through the determination of all separators: a separator $ r (\gx) $ of $ \gs $ in $ V $ follows from the choice of an element of all \[\{r \notin\MFM_{(\gs)}\quad {\rm and}\quad r\in\MFM_{(\gt)} \quad \forall \gt\in V, \gt\neq \gs\}\quad .\]

\section{Application to System biology}
\label{sec:443}

Let us begin with a brief descriptions and reminders of some of the important concepts, as well as of our example data. Once that has been done; we show how this information can be carried out using PDS. DNA repair pathways maintain the integrity of the genome. There are at least 150 different proteins that catalyze DNA repair \cite{Alberts1319}. Conversely, the requirement for DNA repair and genome maintenance in response to therapy implicates DNA repair proteins. This offers novel approaches for tumour selective treatment intended to relieve or heal a disorder. To seed new therapies, researches are needed to explore the detailed consequences of an alteration in each of these repair pathways in general, and in deficient cell lines in particular. Different cell lines provide essential tools to address this need. At first, we selected genes regarding the homologous recombinations pathway \cite{rad}. \\ 

The used information is compared with respect to the evolution of genes in different cell lines, according to the data at different time \cite{art:brca2}. Furthermore, we consider the evolution of each gene independently considering the discretization of their expression \cite{Irizarry01042003}. We assume that given the small changes in gene expressions, Boolean discretization is sufficient for our model. Indeed, we calculate the average expression of each gene over time and we discretize in two Boolean states with respect to this average.
Then, we selected genes that have consistent behaviour compared to the first filter. After this filtering, we have kept the genes common to two cell lines, noted $C_1$ and $C_2$. Among these genes, we chose to keep those who behave differently on each cell lines: we cnsider for simplicity 5 genes in Table \ref{tab:gene5}, for which we compute a model for each cell line, which allows us to compare their evolution, over time. 

\begin{table}[h]
\begin{center}
\begin{tabular}{|c||c|c|c||c|c|c|} \hline
Gene\ cell lines &$C_1$ at 0h &$C_1$ at 24h &$C_1$ at 72h &$C_2$ at 0h &$C_2$ at 24h &$C_2$ at 72h\\ \hline
$g_1$ & 1&1&0&1&0&0\\ \hline
$g_2$ & 1&1&0&1&0&0\\ \hline
$g_3$ & 0&1&1&1&1&0\\ \hline
$g_4$ & 1&0&0&1&0&1\\ \hline
$g_5$ & 1&1&0&0&1&0\\ \hline
\end{tabular}
\caption{Example of 5 genes, common to 2 cell lines $C_1$ and $C_2$ in 3 steps of time.}
\label{tab:gene5}
\end{center}
\end{table}

We can now assume that the number of variables appearing in each component of the PDS will make the difference. What counts is to have a model that sticks most to the realities of data. We obtain two models:  $f_{C_1}$ describes $C_1$ cell lines and $f_{C_2}$ describes $C_2$ cell lines. Simulation and computation were obtained by using SageMath \cite{SteinJoyner2005, sagemath}. 
\begin{eqnarray*}f_{C_1} & = &(0,0, -x_{1} x_{2} x_{3} x_{4} x_{5} + x_{1} x_{2} x_{3} x_{4}, \\
& &-x_{1} x_{2}x_{3} x_{4} x_{5} + x_{1} x_{2} x_{3} x_{5} + x_{1} x_{3} x_{4} x_{5} +
x_{2} x_{3} x_{4} x_{5} - x_{1} x_{3} x_{5} - x_{2} x_{3} x_{5} - x_{3}
x_{4} x_{5} + x_{3} x_{5}, \\
& &-x_{1} x_{2} x_{3} x_{4} x_{5} + x_{1}x_{2} x_{3} x_{4})\\
f_{C_2} & = & (-x_{1} x_{2} x_{3} x_{4} x_{5} + x_{1} x_{2} x_{4} x_{5},  -x_{1} x_{2}
x_{3} x_{4} x_{5} + x_{1} x_{2} x_{4} x_{5}, x_{1} x_{2} x_{3}
x_{5} + x_{1} x_{2} x_{4} x_{5}, \\
& & 0, -x_{1} x_{2} x_{3} x_{4} x_{5}
+ x_{1} x_{2} x_{4} x_{5})
\end{eqnarray*}
Here, $x_1$ is the concentration of $g_1$, $x_2$ the concentration of $g_2$, $x_3$ the concentration of $g_3$, $x_4$ the concentration of $g_4$ and $x_5$ the concentration of $g_5$. Then we deduce a wiring diagram and a state space graph of the given input data. Describing a gene network in terms of Polynomial dynamical system has advantages. First, it describes gene interactions in an explicitly numerical form. Second, these are casual relations between genes: a coefficient $x_i$ in a function $f_j$ determines the effect of gene $i$ on gene $j$. 

\section{Perspective and Conclusion}
\label{conc}

Many biological systems are modeled with discrete models. From the research that has been carried out, it is possible to conclude that effective alternative for Gr\"obner basis when modeling biological systems is realizable. We use classical method based on Lagrange's interpolation. This paper details an approach allowing separators's computation: we present a method based on Galois theory's tools as the {fundamental modules} or {elementary symmetric functions}. \\
The findings have directly practical relevance. We propose an algorithm that performs learning in  Polynomial dynamical systems. For this purpose, we introduce a main theorem for computing algebraic separators with optimal degree, it allows us to compute a finite number of other Polynomial dynamical systems. We also introduce some basic techniques that enable us to compute affine separators. \\

In this context, we presented an analytical method of easily readable expression and easily interpretable specific data. Thus, we obtain a first model in the form of a Polynomial dynamical systems with linear complexity. Many other Polynomial dynamical systems can be determined with our approach optimally. Besides, getting all polynomial dynamical systems from one computed algebraically is also possible: a calculation of the generators of the ideal $I(V)$. The calculation of all the PDS through the determination of all separators. The gain is that very quickly we propose models to bio-informatics and molecular biologists in which they can advance and refine their queries. \\

Clearly, further research will be required on experimental data. Continuing research on this field appears fully justified because of the simplicity of this approach. 
Finally, there is only one parameter $p$ introduced into the model, unlike the continuous model using differential equations which must be added a number of constraints and parameters for successful modelling.



\begin{thebibliography}{99}

\bibitem[1]{CoCoALib}{Abbott, J.}, {Bigatti, A.M.}, {Cocoalib: a {C}++ library for doing {C}omputations in
  {C}ommutative {A}lgebra}. {Available at
  \texttt{http://cocoa.dima.unige.it/cocoalib}}.

\bibitem[2]{Alberts1319}{Alberts, B.},
{Redefining cancer research}. {Science}, {volume} {325},
  {pages} {1319--1319}, {2009}.
\bibitem[3]{Aut2008}
{Benkahla, A.}, {Guizani-Tabbane, L.},
  {Abdeljaoued-Tej, I.}, {BenMiled, S.},
  {Dellagi, K.}. {Systems biology and infectious diseases}.
  {Handbook of Research on Systems Biology Applications in
  Medicine}, {volume} {1}, {pages}{377--402}, {2008}. 
\bibitem[4]{Bortolussi2008}
{Bortolussi, L.}, {Policriti, A.}, {Hybrid Systems and Biology}. {Springer Berlin Heidelberg}, {Berlin,
  Heidelberg}, {pages} {424--448}{2008}..
\bibitem[5]{cinquemani:hal-01399942}
{Cinquemani, E.}, {Donz\'e, A.}.
{Hybrid systems biology}, in:
  {{Proceedings of International Workshop on Hybrid Systems
  Biology}}, {{Springer}}, {Grenoble,
  France}, {2016}.
\bibitem[6]{ed}
{De~Hoon, M.J.L.}, {Imoto, S.},
  {Kobayashi, K.}, {Ogasawara, N.},
  {Miyano, S.}, 
{Inferring gene regulatory networks from time-ordered
  gene expression data of bacillus subtilis using differential equations.}
{Pac Symp Biocomput}, {pages} {17--28}, {2003}.
\bibitem[7]{sagemath}
{Developers, T.S.}, {{S}ageMath, the {S}age {M}athematics {S}oftware
  {S}ystem}.
{{\tt http://www.sagemath.org}}, {2017}.
\bibitem[8]{DBLP:journals/jcb/DimitrovaVML10}
{Dimitrova, E.S.}, {Vera-Licona, P.},
  {McGee, J.}, {Laubenbacher, R.C.}, 
{Discretization of time series data}.
{Journal of Computational Biology},
  {volume} {17}, {pages} {853--868}, {2010}.
\bibitem[9]{Friedman00usingbayesian}
{Friedman, N.}, {Linial, M.},
  {Nachman, I.}, 
{Using bayesian networks to analyze expression data}, 
{Journal of Computational Biology}, 
  {volume} {7}, {pages} {601--620}, {2000}.
\bibitem[10]{M2}
{Grayson, D.R.}, {Stillman, M.E.}, 
{Macaulay2, a software system for research in
  algebraic geometry}. 
{Available at
  {http://www.math.uiuc.edu/Macaulay2/}}, {2000}.
\bibitem[11]{hink}
{Hinkelmann, F.}, {Brandon, M.},
  {Guang, B.}, {McNeill, R.},
  {Blekherman, G.}, {Veliz-Cuba, A.},
  {Laubenbacher, R.}, 
{Adam: Analysis of discrete models of biological
  systems using computer algebra}, 
{BMC Bioinformatics}, {volume} {12},
  {pages} {295}, {2011}.
\bibitem[12]{imoto}
{Imoto, S.}, {Goto, T.},
  {Miyano, S.}, 
{Estimation of genetic networks and functional
  structures between genes by using bayesian networks and nonparametric
  regression.} {Pac Symp Biocomput}, {pages} {175--186}, {2002}.
\bibitem[13]{Irizarry01042003}
{Irizarry, R.A.}, {Hobbs, B.},
  {Collin, F.}, {Beazer‐Barclay, Y.D.},
  {Antonellis, K.J.}, {Scherf, U.},
  {Speed, T.P.}, 
{Exploration, normalization, and summaries of high
  density oligonucleotide array probe level data}.
{Biostatistics}, {volume} {4},
  {pages} {249--264}, 
{10.1093/biostatistics/4.2.249}, {2003}.
\bibitem[14]{jarrah-2006}
{Jarrah, A.S.}, {Laubenbacher, R.},
  {Stigler, B.}, {Stillman, M.}, 
{Reverse-engineering of polynomial dynamical systems}.
{Advances in Applied Mathematics}, 
  {volume} {39, Issue 4}, {pages} {477--489}, {2007}.
\bibitem[15]{doi:10.1093/bioinformatics/bti565}
{Khatri, P.}, {Dr\`aghici, S.}, 
{Ontological analysis of gene expression data: current
  tools, limitations, and open problems}. 
{Bioinformatics} {volume} {21}, 
  {pages} {3587--3595}, {2005}.
\bibitem[16]{Lagrange:1770}
{Lagrange, J.}, {R\'eflexions sur la r\'esolution alg\'ebrique des
  \'equations}, {1770}.
\bibitem[17]{lauben03}
{Laubenbacher, R.}, 
{A computer algebra approach to biological systems},
  in: {booktitle}{Proceedings of the 2003 International Symposium on
  Symbolic and Algebraic Computation}, {ACM},
  {New York, NY, USA}, {2003}.
\bibitem[18]{shortcourse}
{Laubenbacher, R.}, 
{Modeling and Simulation in Applied Mathematics}, 
  {volume} {64} of \textit{{series}{Modeling and
  Simulation of Biological Networks: American Mathematical Society, Short
  Course}}. {AMS Bookstore}, {2007}.
\bibitem[19]{laubenstigler}
{Laubenbacher, R.}, {Stigler, B.},
{A computational algebra approach to the reverse
  engineering of gene regulatory networks}.
{Journal of Theoretical Biology}, 
  {volume} {229}, {pages} {523 -- 537}, {2004}.
\bibitem[20]{Paroni2016}
{Paroni, A.}, {Graudenzi, A.},
  {Caravagna, G.}, {Damiani, C.},
  {Mauri, G.}, {Antoniotti, M.}, 
{Cabernet: a cytoscape app for augmented boolean
  models of gene regulatory networks}.
{BMC Bioinformatics}, {volume} {17},
  {pages} {64},  {2016}.
\bibitem[21]{Schilstra2004}
{Schilstra, M.}, {Bolouri, H.}, 
{Models of Genetic Regulatory Networks}.
  {Springer Berlin Heidelberg}, {Berlin,
  Heidelberg}, 
{pages} {149--159}, {2004}.
\bibitem[22]{SteinJoyner2005}
{Stein, W.}, {Joyner, D.}, 
{{SAGE}: System for algebra and geometry
  experimentation}.
{ACM SIGSAM Bulletin}, {volume} {39},
  {pages} {61--64}, {2005}.
\bibitem[23]{shortcourseStigler}
{Stigler, B.}, 
{Polynomial Dynamical Systems in Systems Biology}, 
  {volume} {64} of \textit{{series}{Modeling and
  Simulation of Biological Networks: American Mathematical Society, Short
  Course}}, 
{AMS Bookstore}, {2006}.
\bibitem[24]{Tchebotarev:50}
{Tchebotarev, N.}, 
{Gr\"{u}ndz\"{u}ge des Galois'shen Theorie}.
{P. Noordhoff}, {1950}.
\bibitem[25]{thomas}
{Thomas, R.}, 
{Kinetic logic : a boolean approach to the analysis of
  complex regulatory systems}, {volume} {29}, 
{Lecture Notes in Biomathematics}, {1979}.
\bibitem[26]{thomas:hal-00087681}
{Thomas, R.}, {D'Ari, R.}, 
{{Biological Feedback}}, 
{{CRC Press, Inc.}}, {1990}.
\bibitem[27]{rad}
{Tomasz, S.}, 
{Oncogenic tyrosine kinases and the dna-damage
  response}, 
{Nature Reviews Cancer}, {volume} {2},
  {pages} {351-- 360}, {2002}.
\bibitem[28]{val:sym}
{Valibouze, A.}, 
{Symbolic computation with symmetric polynomials, an
  extension to {M}acsyma}, in: {Computers and Mathematics}, {Springer-Verlag, New York Berlin}, 
  {pages} {308--320}, {1989}.
\bibitem[29]{Valibouze:08}
{Valibouze, A.}, 
{Sur les relations entre les racines d'un polyn\^ome}.
{Acta Arithmetica}, {volume} {131.1},
  {pages} {1--27}, {2008}.
\bibitem[30]{iaa}
{Valibouze, A.}, {Abdeljaoued, I.},
  {Kahla, A.B.}, 
{Galoisian separators for biological systems}, in: {Mathematics Algorithms Proofs - Formalization of
  Mathematics}, {Monastir, Tunisia}, {2009}.
\bibitem[31]{art:brca2}
{Valsecchi, M.}, {Díaz-Cantón, E.}, {Vega, M.}, {Littman, S.}, {Recent treatment advances and novel therapies in
  pancreas cancer: A review}, {Journal of Gastrointestinal Cancer}, 
  {volume} {45}, {pages} {190--201}, {2014}.
\bibitem[32]{jarrah-2005}{Vastani, H.}, {Jarrah, A.S.},
  {Laubenbacher, R.}. {package, {D}iscrete {V}isualizer of {D}ynamics}, 2005.

\end{thebibliography}

\end{document}